\documentclass[runningheads]{llncs}
\usepackage{amssymb}
\usepackage{amsmath}
\usepackage{comment}
\usepackage[color=gray!27]{todonotes}
\usepackage[utf8]{inputenc}
\usepackage{inputenc}
\usepackage{shuffle}
\usepackage{multirow}
\usepackage{listings}
\usepackage{mathabx}
\usepackage{hyperref}

\usepackage{shuffle}

\usepackage{tikz}
\usetikzlibrary{positioning,shadows,arrows}
\usetikzlibrary{arrows,automata,positioning,calc}

\usepackage{diagbox}
\usepackage{slashbox}
\usepackage{tikz}
\usetikzlibrary{matrix}

\usepackage{apxproof}
\theoremstyle{plain}
\newtheoremrep{theorem}{Theorem}
\newtheoremrep{proposition}[theorem]{Proposition}
\newtheoremrep{lemma}[theorem]{Lemma}
\newtheoremrep{claim}[theorem]{Claim}
\newtheoremrep{conjecture}[theorem]{Conjecture}
\newtheoremrep{corollary}[theorem]{Corollary}
\newtheoremrep{definition}[theorem]{Definition}

\newenvironment{claiminproof}[1]{\medskip\par\noindent\underline{Claim:}\space#1}{}
\newenvironment{claimproof}[1]{\begin{quote}\par\noindent\emph{Proof of the Claim:}\space#1}{[\emph{End, Proof of the Claim}]\end{quote}}

\DeclareFontFamily{U}{bigshuffle}{}
\DeclareFontShape{U}{bigshuffle}{m}{n}{
  <5-8> s*[1.7] shuffle7
  <8->  s*[1.7] shuffle10
}{}
\DeclareSymbolFont{BigShuffle}{U}{bigshuffle}{m}{n}
\DeclareMathSymbol\bigshuffle{\mathop}{BigShuffle}{"001}
\DeclareMathSymbol\bigcshuffle{\mathop}{BigShuffle}{"002}

\newcommand{\NP}{\textsf{NP}}
\newcommand{\PSPACE}{\textsf{PSPACE}}

\newcommand{\PTIME}{\textsf{P}}

\begin{document}
\title{Constrained Synchronization and Subset Synchronization Problems for Weakly Acyclic Automata}
\titlerunning{Constrained Synchronization and Weakly Acyclic Automata}

%
%
\author{Stefan Hoffmann\orcidID{0000-0002-7866-075X}}
\authorrunning{S. Hoffmann}
%
\institute{Informatikwissenschaften, FB IV, 
  Universit\"at Trier,  Universitätsring 15, 54296~Trier, Germany, 
  \email{hoffmanns@informatik.uni-trier.de}}
\maketitle              
\begin{abstract}
 We investigate the constrained synchronization problem for weakly acyclic, or partially ordered,
 input automata. We show that, for input automata of this type, the problem is always
 in $\NP$.
 Furthermore, we give a full classification of the realizable complexities for constraint
 automata with at most two states and over a ternary alphabet.
 We find that most constrained problems that are \PSPACE-complete in general
 become \NP-complete. However, there also 
 exist constrained problems that are \PSPACE-complete in the general setting
 but become polynomial time solvable when considered for weakly acyclic input automata.
 We also investigate
 two problems related to subset synchronization, namely if there exists
 a word mapping all states into a given target subset of states, and
 if there exists a word mapping one subset into another. Both problems
 are $\PSPACE$-complete in general, but in our setting the former is polynomial time solvable and the latter is $\NP$-complete.
 %

\end{abstract}

\section{Introduction}
\label{sec:introduction}



A deterministic semi-automaton is synchronizing if it admits a reset word, i.e., a word which leads to a definite
state, regardless of the starting state. This notion has a wide range of applications, from software testing, circuit synthesis, communication engineering and the like, see~\cite{San2005,Vol2008}.  The famous \v{C}ern\'y conjecture~\cite{Cer64}
states that a minimal length synchronizing word, for an $n$-state automaton, has length
at most $(n-1)^2$. 
We refer to the mentioned survey articles~\cite{San2005,Vol2008} for details\footnote{A new and updated survey article (in Russian) is currently in preparation by Mikhail V. Volkov~\cite{VolkovSurvey}.}.

Due to its importance, the notion of synchronization has undergone a range of generalizations and variations
for other automata models.
In some  generalizations, related to partial automata~\cite{Martyugin12}, only certain paths, or input words, are allowed (namely those for which the input automaton is defined). 

In~\cite{Gusev:2012}
the notion of constrained synchronization was 
introduced in connection with a reduction procedure
for synchronizing automata.
The paper \cite{DBLP:conf/mfcs/FernauGHHVW19} introduced the computational problem of constrained 
synchronization. In this problem, we search for a synchronizing word coming from a specific subset of allowed
input sequences. For further motivation and applications we refer to the aforementioned paper \cite{DBLP:conf/mfcs/FernauGHHVW19}.
In this paper, a complete analysis of the complexity landscape when the constraint language is given by 
small partial automata with up to two states and an at most ternary alphabet was done.
It is natural to extend this result to other language classes, or
even to give a complete classification of all the complexity classes that could arise.
For commutative regular constraint languages, a full classification of the realizable
complexities was given in~\cite{DBLP:conf/cocoon/Hoffmann20}.
In~\cite{DBLP:conf/ictcs/Hoffmann20}, it was shown that for polycyclic constraint languages, the problem is always in~$\NP$.

Let us mention that restricting the solution space by a regular language
has also been applied in other areas, for example to topological sorting~\cite{DBLP:conf/icalp/AmarilliP18},
solving word equations~\cite{Diekert98TR,DBLP:journals/iandc/DiekertGH05}, constraint programming~\cite{DBLP:conf/cp/Pesant04}, or
shortest path problems~\cite{DBLP:journals/ipl/Romeuf88}.
The road coloring problem asks for a labeling of a given graph  such that a synchronizing
automaton results. A closely related problem to our problem of constrained synchronization is to restrict the possible labeling(s), and
this problem was investigated in~\cite{DBLP:journals/jcss/VorelR19}.

In~\cite{DBLP:conf/mfcs/FernauGHHVW19} it was shown that we can realize $\PSPACE$-complete, $\NP$-complete
or polynomial time solvable constrained problems by appropriately choosing 
a constraint language.
Investigating the reductions from~\cite{DBLP:conf/mfcs/FernauGHHVW19},
we see that most reductions yield automata with a sink state, which then must be the unique synchronizing state.
Hence, we can conclude
that we can realize these complexities with this type of input automaton.

Contrary, for example, unary automata are synchronizing only if they admit no non-trivial cycle, i.e., 
only a single self-loop. In this case, we can easily decide synchronizability
for any constraint language in polynomial time.
Hence, for these simple types of automata, the complexity drops considerably.
So, a natural question is, if we restrict the class of input automata, what complexities
are realizable?

Here, we will investigate this question for the class of weakly acyclic input automata.
These are automata such that the transition relation induces a partial order on the state sets.
We will show that for this class, the constrained synchronization problem
is always in $\NP$.
Then, in the spirit of the work~\cite{DBLP:conf/mfcs/FernauGHHVW19},
we will give a full classification of the complexity landscape for constraint automata
with up to three states and a ternary alphabet.
Compared with the classification result from~\cite{DBLP:conf/mfcs/FernauGHHVW19},
we find that most problems that are $\PSPACE$-complete in general will become $\NP$-complete.
However, a few, in general $\PSPACE$-complete, cases become polynomial time solvable
for weakly acyclic input automata.

Related synchronization problems for weakly acyclic automata were previously investigated 
in~\cite{Ryzhikov19a}. For example, in~\cite{Ryzhikov19a},
it was shown that the problem to decide if a given subset of states could be mapped to a
single state, a problem $\PSPACE$-complete for general automata~\cite{DBLP:journals/jcss/BerlinkovFS21,DBLP:journals/ipl/Rystsov83}, 
is $\NP$-complete for weakly acyclic automata.

Furthermore, we investigate two problems
related to subset synchronization, namely the problem if we can map the whole state set into a given target set by some word, 
and if we can map any given starting set into another target set.
Both problems are $\PSPACE$-complete in general~\cite{DBLP:journals/jcss/BerlinkovFS21,DBLP:journals/cc/BlondinKM16,DBLP:conf/ictcs/Hoffmann20,DBLP:journals/jcss/LuksM88,DBLP:journals/ipl/Rystsov83,San2005}. However, for weakly acyclic automata
the former becomes polynomial time solvable, as we will show here, and the latter becomes $\NP$-complete.

Similar subset synchronization problems, for general, strongly connected and synchronizing automata, were investigated in~\cite{DBLP:journals/jcss/BerlinkovFS21}.

Weakly acyclic automata are also known as partially ordered automata~\cite{DBLP:journals/jcss/BrzozowskiF80},
or acyclic automata~\cite{DBLP:conf/wia/JiraskovaM12}.
As shown in~\cite{DBLP:journals/jcss/BrzozowskiF80}, the languages recognized
by weakly acyclic automata are precisely the languages recognized by $\mathcal R$-trivial monoids.

\section{Preliminaries}
\label{sec:preliminaries}

By $\Sigma$ we denote a finite set of symbols, also called an \emph{alphabet}.
By $\Sigma^*$ we denote the set of all \emph{words} over $\Sigma$, i.e., finite sequences
with the concatenation operation. The \emph{empty word} is denoted by $\varepsilon$.
A \emph{language} $L$ is a subset $L \subseteq \Sigma^*$.

A \emph{partial deterministic finite automaton (PDFA)}
is denoted by a quintuple $\mathcal A =(\Sigma, Q, \delta, q_0, F)$,
where $Q$ is a \emph{finite set of states}, $\Sigma$ the \emph{input alphabet},
$\delta : Q \times \Sigma \to Q$ is a \emph{partial transition function},
$q_0$ the
\emph{start state} and $F \subseteq Q$ the set of \emph{final states}.
An automaton $\mathcal A = (\Sigma, Q, \delta, q_0, F)$
is called \emph{complete}, if $\delta$ is a total
function, i.e., $\delta(q, x)$ is defined for any $q \in Q$
and $x \in \Sigma$.

In the usual way, the transition function $\delta$
can be extended to a function $\hat \delta : Q \times \Sigma^* \to Q$ by setting, for $q \in Q$, $u \in \Sigma^*$
and $x \in \Sigma$, $\hat \delta(q, \varepsilon) = q$
and $\hat \delta(q, ux) = \delta(\hat \delta(q, u), x)$.
In the following, we will drop the distinction between $\delta$ and $\hat \delta$
and will denote both functions simply by $\delta$.

For $S \subseteq Q$ and $u \in \Sigma^*$, we set $\delta(S, u) = \{ \delta(s, u) \mid s \in S \mbox{ and } \delta(s, u) \mbox{ is defined } \}$
and
$\delta^{-1}(S, u) = \{ q \in Q \mid \delta(q, u) \mbox{ is defined and } \delta(q, u) \in S \}$.
For $q \in Q$ and $u \in \Sigma^*$, we set $\delta^{-1}(q, u) = \delta^{-1}(\{q\}, u)$.

The language \emph{recognized} by $\mathcal A$
is $L(\mathcal A) = \{ u \in \Sigma^* \mid \delta(q_0, u) \in F\}$.

We say that $q \in Q$ is \emph{reachable} from $p \in Q$ (in $\mathcal A$)
if there exists a word $u \in \Sigma^*$ such that $\delta(p, u) = q$.

For $\mathcal A = (\Sigma, Q, \delta, q_0, F)$ and $\Gamma \subseteq \Sigma$,
by $\mathcal A_{|\Gamma} = (\Gamma, Q, \delta_{|\Gamma}, q_0, F)$
we denote the automaton $\mathcal A$ \emph{restricted to the subalphabet} $\Gamma$, i.e.,
$\delta_{|\Gamma} : Q \times \Gamma \to Q$ with $\delta_{|\Gamma}(q, x) = \delta(q, x)$
for $q \in Q$ and $x \in \Gamma$.

We say a letter $x \in \Sigma$ induces
a \emph{self-loop} at a state $q \in Q$,
if $\delta(q, x) = q$.

A state $s \in Q$ is called a \emph{sink state},
if every letter induces a self-loop at it, i.e.,
$\delta(q, x) = q$ for any $x \in \Sigma$.

An automaton $\mathcal A = (\Sigma, Q, \delta, q_0, F)$
is called \emph{weakly acyclic}, if it is complete
and for any $q \in Q$ and $u \in \Sigma^* \setminus \{\varepsilon\}$,
if $\delta(q, u) = q$, then $\delta(q, x) = q$
for any letter $x$ appearing in $u$, i.e., the simple\footnote{A cycle is simple
if it only involves distinct states~\cite{Ryzhikov19a}.} cycles are self-loops. Equivalently, the reachability 
relation is a partial order. Here, we say a state $q$ is larger than another state $p$, if
$q$ is reachable from $p$ in $\mathcal A$. A state
in a weakly acyclic automaton
is called \emph{maximal}, if it is maximal
with respect to this partial order.
Note that here, we require
weakly acyclic automata to be complete.
This is in concordance with~\cite{Ryzhikov19a}.
However, partially ordered automata are sometimes allowed to be partial in the literature~\cite{DBLP:journals/iandc/KrotzschMT17}.
Equivalently, an automaton is weakly acyclic
if and only if there exists an ordering $q_1,\ldots, q_n$
of its states such that if $\delta(q_i, x) = q_j$
for some letter $x \in \Sigma$,
then $i \le j$, i.e., we can \emph{topologically sort} the states.

A \emph{semi-automaton} $\mathcal A = (\Sigma, Q, \delta)$ 
is a finite complete automaton without a specified start state
and with no specified set of final states.
Every notion defined for complete automata
that does not explicitly use the start state and the set of final states
is also defined in the same way for semi-automata.
For example, being weakly acyclic.
When the context is clear, we call both finite automata and semi-automata simply \emph{automata}.

A complete automaton $\mathcal A$ is called \emph{synchronizing} if there exists a word $w \in \Sigma^*$ with $|\delta(Q, w)| = 1$. In this case, we call $w$ a \emph{synchronizing word} for $\mathcal A$.
We call a state $q\in Q$ with $\delta(Q, w)=\{q\}$ for some synchronizing word $w\in \Sigma^*$ a \emph{synchronizing state}.

For a fixed PDFA $\mathcal B = (\Sigma, P, \mu, p_0, F)$,
we define the \emph{constrained synchronization problem}:

\begin{quote}
\begin{definition}
 \textsc{$L(\mathcal B)$-Constr-Sync}\\
\emph{Input}: Deterministic semi-automaton $\mathcal A = (\Sigma, Q, \delta)$.\\
\emph{Question}: Is there a synchronizing word $w$ for $\mathcal A$ with  $w \in L(\mathcal B)$?
\end{definition}
\end{quote}

The automaton $\mathcal B$ will be called the \emph{constraint automaton}.
If an automaton~$\mathcal A$ is a yes-instance of \textsc{$L(\mathcal B)$-Constr-Sync} we call $\mathcal A$ \emph{synchronizing with respect to~$\mathcal{B}$}. Occasionally,
we do not specify $\mathcal{B}$ and rather talk about \textsc{$L$-Constr-Sync}.
The unrestricted synchronization problem, i.e., $\Sigma^*$\textsc{-Constr-Sync}
in our notation, is in $\PTIME$~\cite{Vol2008}.
We are going to investigate this problem
for weakly acyclic input automata only.

\begin{quote}
\begin{definition}
 \textsc{$L(\mathcal B)$-WAA-Constr-Sync}\\
\emph{Input}: Weakly acyclic semi-automaton $\mathcal A = (\Sigma, Q, \delta)$.\\
\emph{Question}: Is there a synchronizing word $w$ for $\mathcal A$ with  $w \in L(\mathcal B)$?
\end{definition}
\end{quote}

We assume the reader to have some basic knowledge in computational complexity theory and formal language theory, as contained, e.g., in~\cite{HopUll79}. For instance, we make use of  regular expressions to describe languages.
And we make use of complexity classes like $\PTIME$, $\NP$, or $\PSPACE$. 
The following was shown in~\cite{DBLP:conf/mfcs/FernauGHHVW19}.

\begin{theorem}[\cite{DBLP:conf/mfcs/FernauGHHVW19}]
\label{thm:classification_MFCS_paper}
 Let $\mathcal B = (\Sigma, P, \mu, p_0, F)$
 be a PDFA.
 If $|P|\le 1$ or $|P| = 2$ and $|\Sigma|\le 2$, then $L(\mathcal B)\textsc{-Constr-Sync} \in \PTIME$.
 For $|P| = 2$ with $|\Sigma| = 3$, up to symmetry by renaming of the letters,
 $L(\mathcal B)\textsc{-Constr-Sync}$
 is $\PSPACE$-complete precisely in the following cases for $L(\mathcal B)$:
 $$
  \begin{array}{llll}
    a(b+c)^*        & (a+b+c)(a+b)^*  & (a+b)(a+c)^* & (a+b)^*c \\
    (a+b)^*ca^*     & (a+b)^*c(a+b)^* & (a+b)^*cc^*  & a^*b(a+c)^* \\
    a^*(b+c)(a+b)^* & a^*b(b+c)^*     & (a+b)^*c(b+c)^* & a^*(b+c)(b+c)^*
  \end{array}
 $$
 and polynomial time solvable in all other cases.
\end{theorem}

In weakly acyclic automata, maximal
states, sink states and synchronizing
states are related as stated in the next lemmata.

\begin{lemmarep} 
\label{lem:waa_sink_equals_max}
 In a weakly acyclic automaton\footnote{Recall that here, weakly acyclic automata are always complete. For partial automata such that the reachability relation is a partial order, this does not have to be true.} a state is maximal
 if and only if it is a sink state.
\end{lemmarep}
\begin{proof}
 As from a sink state no other state is reachable, it could not have any
 proper successor states, hence is maximal.
 Conversely, if a state is maximal, then by definition no other state is reachable
 from it, hence, every outgoing transition has to go back to this state, i.e.,
 induces a self-loop.~\qed
\end{proof}

%
%
\begin{lemmarep} 
\label{lem:sync_waa_has_sink_state}
 Let $\mathcal A = (\Sigma, Q, \delta)$
 be a weakly acyclic automaton.
 If $\mathcal A$ is synchronizing, then the synchronizing
 state must be a unique sink state in $\mathcal A$
 that is reachable from every other state and, conversely,
 such a state is a synchronizing state.
\end{lemmarep}
\begin{proof}
 Let $\mathcal A = (\Sigma, Q, \delta)$ be weakly acyclic
 and $w \in \Sigma^*$ be
 such that $\delta(Q, w) = \{ q\}$ for some $q \in Q$.
 Hence, $q$ is reachable from every other state and so must be maximal.
 By Lemma~\ref{lem:waa_sink_equals_max}, $q$ is a sink state.
 Conversely, if we have a sink state $s \in Q$ 
 reachable from every other state $q \in Q$ by a word $w_q$,
 we can construct a synchronizing word.
 We can suppose $Q$ has more than two states, for otherwise the problem is trivial.
 Set $w_1 = w_{q}$ and $S_1 = \delta(Q, w_1)$ for some $q \in Q \setminus \{s\}$.
 Then, inductively, let $i > 1$ and, if $|S_{i-1}| > 1$,
 choose $q \in S_{i-1} \setminus \{s\}$ and set $w_i = w_{i-1} w_{q}$.
 As $q$ and $s$ are mapped to $s \in S_{i-1}$, in this case $|S_i| < |S_{i-1}|$.
 So, after at most $|Q| - 1$ many steps, for some $i \in \{1,\ldots, |Q|\}$
 we must have $S_i = \{s\}$ and $\delta(Q, w_i) = \{s\}$.~\qed
\end{proof}

\begin{toappendix}
\begin{remark} 
Note that the procedure used in the proof of Lemma~\ref{lem:sync_waa_has_sink_state}
works for any automaton with a sink state reachable from any other state.
\end{remark}
\end{toappendix}

With Lemma~\ref{lem:sync_waa_has_sink_state}, we can test if a given
weakly acyclic automaton is synchronizing.
First, check every state
if it is a sink state. If we have found a unique sink state,
then do a breadth-first search from this sink state by traversing the transitions
in the reverse direction. This gives a better algorithm than the general algorithm, which runs in time $O(|\Sigma||Q|^2)$, see~\cite{Vol2008}.

\begin{corollary}
 For weakly acyclic automata we can decide in time $O(|\Sigma||Q| + |Q|)$
 if it is synchronizing.
\end{corollary}



\section{Constrained Synchronization of Weakly Acyclic Automata}

In general, for any constraint automaton, the constrained synchronization problem
is always in $\PSPACE$, see~\cite{DBLP:conf/mfcs/FernauGHHVW19}.
Here, we show that for weaky acyclic input automata, the constrained
synchronization problem is always in $\NP$. First, we establish
a bound on the size of a shortest synchronizing word, which
directly yields containment in $\NP$ as we have a polynomially bounded certificate
which could be verified in polynomial time.

\begin{proposition}
\label{prop:length_constr_sync_word_in_waa}
 Let $\mathcal A$ be a weakly acyclic automaton with $n$ states and $\mathcal B = (\Sigma, P, \mu, p_0, F)$ be a fixed PDFA.
 Then, a shortest synchronizing word $w \in L(\mathcal B)$ for $\mathcal A$
 has length at most $|P|\binom{n}{2}$. 
\end{proposition}
\begin{proof} 
 %
 Let $q_1 ,\ldots, q_n$ be a topological sorting of the states of $\mathcal A$.
 We represent the situation after reading a word $u \in \Sigma^*$, i.e., the set $\delta(Q, u)$,
 by a tuple $(i_1, \ldots, i_n) \in \mathbb \{1,\ldots,n\}^n$,
 where $i_j$ is the index of $\delta(q_j, u)$ in the topological sorting, i.e.,
 $\delta(q_j, u) = q_{i_j}$.
 Then, $u \in \Sigma^*$ is synchronizing if and only if
 the corresponding tuple is $(n, \ldots, n)$.
 The starting tuple is $(1, \ldots, n)$.
 For $(i_1, \ldots, i_n), (j_1, \ldots, j_n) \in \{1,\ldots,n\}^n$
 we write $(i_1, \ldots, i_n) < (j_1, \ldots, j_n)$
 if, for all $r \in \{1,\ldots,n\}$, we have $i_r \le j_r$
 and there exists at least one $s \in \{1,\ldots, n \}$
 such that $i_s < j_s$.

 Let $w = x_1 \cdots x_m \in L(\mathcal B)$
 with $x_i \in \Sigma$ for $i \in \{1,\ldots,m\}$. 
 Then, set $S_i = \delta(Q, x_1 \cdots x_i)$ and $S_0 = Q$.
 Suppose $S_{i+ |P|} = S_i$
 for some $i \in \{0, 1,\ldots, n\}$. 
 Then, as $\mathcal A$ is weakly acyclic\footnote{More generally, it is also easy to see that in weakly acyclic automata, no word can induce a non-trivial permutation of a subset of states.},
 for the word $u = x_{i+1} \cdots x_{i + |P|}$
 we have $\delta(q, u) = q$ for any $q \in S_i$
 and, as it has length $|P|$, it induces a loop in the constraint automaton $\mathcal B$.
 So, we can replace this factor of $w$
 by a shorter word $v \in \Sigma^*$ of length less than $|P|$
 that yields the same result, i.e.,
 $
  S_{i + |P|} = \delta(Q, x_1\cdots x_i v)
 $
 and $x_1\cdots x_i v x_{i + |P| + 1} \cdots x_n \in L(\mathcal B)$.

 Now, suppose $w = x_1 \cdots x_m \in L(\mathcal B)$ is a shortest synchronizing word for $\mathcal A$. 
 By the previous paragraph, we can suppose 
 $S_{i + |P|} \ne S_i$ for any $i \in \{1,\ldots, n\}$.
 As $\mathcal A$ is weakly acyclic, and we can only move forward in the topological sorting,
 if $\delta(Q, u) \ne \delta(Q, uv)$,
 then for the tuple $(i_1, \ldots, i_n)$ corresponding
 to $\delta(Q, u)$ and for the tuple $(j_1, \ldots, j_n)$ for $\delta(Q, uv)$
 we have $(i_1, \ldots, i_n) < (j_1, \ldots, j_n)$. Note that we have equality
 if and only if $\delta(Q, u) = \delta(Q, uv)$.
 As we start with $(1, \ldots, n)$
 and want to reach $(n, \ldots, n)$, 
 we have to increase at least $n - 1$ times the first entry, $n - 2$ times the second
 and so on.
 Now, by the previous reasoning, every $|P|$ symbols we can suppose we increase
 some component. Combining these observations yields that a shortest
 synchronizing word has length at most
 \[
  |P| \cdot \left( (n - 1) + (n - 2) + \ldots + 1 \right) = |P| \cdot \binom{n}{2}.
 \]
 This finishes the proof.~\qed
\end{proof}

With Proposition~\ref{prop:length_constr_sync_word_in_waa} we can conclude that
for weakly acyclic input automata, the constrained synchronization problem is always in $\NP$.

\begin{theoremrep}
\label{thm:constraint_sync_weakly_acyclic_in_NP}
 For weakly acyclic input automata and an arbitrary constraint automaton,
 the constrained synchronization problem
 is in $\NP$.
\end{theoremrep}
\begin{proof}
 By Proposition~\ref{prop:length_constr_sync_word_in_waa} we can guess a shortest synchronizing word in $L(\mathcal B)$
 of polynomial length.
 Verifying that such a word is indeed synchronizing could be done in polynomial time.~\qed
\end{proof}

\section{Subset Synchronization Problems}

Here, we will investigate the followig problems 
from~\cite{DBLP:journals/jcss/BerlinkovFS21,DBLP:journals/cc/BlondinKM16,DBLP:conf/ictcs/Hoffmann20,DBLP:journals/jcss/LuksM88,DBLP:journals/ipl/Rystsov83,San2005,DBLP:journals/ijfcs/Vorel16}
for weakly acyclic input automata.

\medskip
\noindent
\scalebox{.98}{\begin{minipage}{0.48\textwidth}
\begin{definition}
	\textsc{Sync-From-Subset}\\
	\emph{Input}: $A = (\Sigma, Q, \delta)$ and $S \subseteq Q$.\\
	\emph{Question}: Is there a word $w$ with $|\delta(S, w)| = 1$?
\end{definition}\end{minipage}}\hfill
\scalebox{.98}{\begin{minipage}{0.48\textwidth}
\begin{definition}
	\textsc{Sync-Into-Subset}\\
	\emph{Input}: $A = (\Sigma, Q, \delta)$ and $S \subseteq Q$.\\
	\emph{Question}: Is there a word $w$ with $\delta(Q, w) \subseteq S$?
\end{definition}
\end{minipage}}
\medskip


\begin{definition}\label{def:problem_set_transporter}
  \textsc{SetTransporter}\\
  \emph{Input:} $\mathcal A = (\Sigma, Q, \delta)$ and two subsets $S, T \subseteq Q$.\\
  \emph{Question:} Is there a word $w \in \Sigma^*$ such that $\delta(S, w) \subseteq T$?
\end{definition}


These problems are $\PSPACE$-complete in general~\cite{DBLP:journals/jcss/BerlinkovFS21,DBLP:journals/cc/BlondinKM16,DBLP:journals/ipl/Rystsov83,San2005} for at least binary alphabets.
In~\cite{Ryzhikov19a} it was shown that \textsc{Sync-From-Subset}
is $\NP$-complete for weakly acyclic input automata. 
Interestingly, for weakly acyclic input automata,
the complexity of \textsc{Sync-Into-Subset}
drops considerably. Namely, we could solve the problem in polynomial time.
Hence, the ability to have transitions that go backward seems to be essential
to get hardness above polynomial time solvability for this problem.

%
%
%
%
%
%


\begin{theoremrep}
\label{thm:sync_into_subset_waa_in_P}
 The problem $\textsc{Sync-Into-Subset}$
 is polynomial time solvable for weakly acyclic input automata.
 More generally\footnote{This more general formulation was pointed out by
 an anonymous referee.}, given $S, T \subseteq Q$ such that~$S$
 contains all maximal states reachable from~$S$,
 the existence of a word $w \in \Sigma^*$
 such that $\delta(S, w) \subseteq T$
 could be decided in polynomial time.
\end{theoremrep}
\begin{proof}
 We show the more general claim, the implication for $\textsc{Sync-Into-Subset}$
 is then implied by setting $S = Q$.
 Let $\mathcal A = (\Sigma, Q, \delta)$ be a weakly acyclic automaton.
 Let $R \subseteq S$ be the set of maximal states reachable from $S$ in $\mathcal A$
 and suppose $R \subseteq S$ and $T \subseteq Q$. Set $n = |Q|$.
 
 \begin{claiminproof}
 There exists $w \in \Sigma^*$
 such that $\delta(S, w) \subseteq T$
 if and only if $R \subseteq T$.
 \end{claiminproof}
 \begin{claimproof}
  By Lemma~\ref{lem:waa_sink_equals_max},
  the set $R$ only contains sink states.
  So, for any $w \in \Sigma^*$, we have $R \subseteq \delta(S, w)$.
  Hence, if $\delta(S, w) \subseteq T$, then $R \subseteq T$.
  Conversely, suppose $R \subseteq T$.
  Let $w_{\Sigma}$ contain every symbol from the input alphabet~$\Sigma$
  exactly once, in any order. Apart from the states in $R$, for every other state $q$ in $S$,
  there is a state $q' \ne q$ that can be reached 
  from $q$ by reading one symbol.
  As self-loops are the only cycles in weakly acyclic automata,
  we have $\delta(S, w_{\Sigma}^{n - 1}) \subseteq R$.
  Hence, we find that $\delta(S, w_{\Sigma}^{n - 1}) \subseteq T$.
 \end{claimproof}
 Hence, we only have to check if $R \subseteq T$, which could
 be done in polynomial time, as $R$ is easily computable.~\qed
\end{proof}

Not surprisingly, as \textsc{Sync-From-Subset} is $\NP$-complete~\cite{Ryzhikov19a} for at least binary alphabets, \textsc{SetTransporter} is $\NP$-complete
for at least binary alphabets.

\begin{theorem}\label{thm:set_transporter_weakly_acyclic}
 $\textsc{SetTransporter}$
 is $\NP$-complete for weakly acyclic input automata
 when the alphabet is fixed but contains at least two distinct letters.
\end{theorem}
\begin{proof}
 For containment in $\NP$, suppose $(\mathcal A, S, T)$
 with $\mathcal A = (\Sigma, Q, \delta)$, $S, T \subseteq Q$,
 is an instance of $\textsc{SetTransporter}$ with $\mathcal A$ being weakly acyclic.
 Let $a,b \notin \Sigma$ be two new symbols and $s_f \notin Q$
 a new state. We can suppose $S, T$
 are non-empty, for otherwise, if $S = \emptyset$
 we have a trivial solution and if $S$ is non-empty and $T = \emptyset$ we have no solution at all.
 Then, construct $\mathcal A' = (\Sigma \cup \{a,b\}, Q \cup \{s_f\}, \delta')$
 with, for $q \in Q$ and $x \in \Sigma \cup \{a,b\}$,
 \[
  \delta'(q, x) = \left\{ 
  \begin{array}{ll}
   \delta(q, x) & \mbox{if } x \in \Sigma; \\ 
   s_f          & \mbox{if } x = a \mbox{ and } q \notin S; \\ 
   s_f          & \mbox{if } x = b \mbox{ and } q \in T; \\
   q            & \mbox{otherwise.}
  \end{array}
  \right.
 \]
 and $\delta'(s_f, x) = s_f$ for any $x \in \Sigma \cup \{a,b\}$.
 Note that $\delta'(Q\cup\{s_f\}, a) = S \cup \{s_f\}$, $\delta'(q, b) = s_f$ for $q \in Q$ if and only if $q \in T$ and that $\mathcal A'$ is weakly acyclic as we have only added self-loops
 or transitions going into the sink state $s_f$.
 Then, there exists $w \in \Sigma^*$
 such that $\delta(S, w) \subseteq T$ in $\mathcal A$
 if and only if $\delta'(Q, awb) = \{s_f\}$ in $\mathcal A'$.
 So, we have reduced the original problem to the problem
 to decide if $\mathcal A'$ has a synchronizing word for the constraint 
 language $a\Sigma^*b$. By Theorem~\ref{thm:constraint_sync_weakly_acyclic_in_NP},
 the last problem is in $\NP$.

 For \NP-hardness, we can use the same reduction
 as used in~\cite[Theorem 4]{Ryzhikov19a} to show \NP-hardness
 of $\textsc{Sync-From-Subset}$ with the same set $S$
 but setting $T = \{f\}$, where $f$ is the sink state used 
 in the reduction from~\cite{Ryzhikov19a}.\qed 
\end{proof}



In~\cite{DBLP:conf/ictcs/Hoffmann20}, it was shown that \textsc{SetTransporter}
is $\NP$-complete for general unary automata. For unary weakly acyclic automata,
the problem is in $\PTIME$.

\begin{propositionrep}
 If $|\Sigma| = 1$, then \textsc{SetTransporter}
 is in $\PTIME$ for weakly acyclic input automata.
\end{propositionrep}
\begin{proof}
 If $\mathcal A = (\{a\}, Q, \delta)$ is weakly acyclic over a unary alphabet, then
 it is a collection of paths that end in single state with a self-loops, where multiple paths
 could end in the same state labeled by a self-loop.
 Hence, for a given set $S \subset Q$, the states in $\delta(Q, a^{|Q|-1})$
 are precisely those states that are labeled by self-loops.
 Hence, $\delta(Q, a^i) = \delta(Q, a^{|Q|-1})$ for all $i \ge |Q|$
 and we only need to test the state sets $\delta(Q, a^i)$ for $i \in \{1,\ldots, |Q| - 1\}$
 if they equal a given target set $T \subseteq Q$. All these operations
 could be performed in polynomial time.~\qed
\end{proof}

\section{Constraint Automata with Two States and at most Three Letters}

%

Here, we give a complete classification of the complexity landscape
of the constraint synchronization problem with weakly acyclic automata
as input automata and when the constraint is given by an at most two state
PDFA over an at most ternary alphabet.

For our $\NP$-hardness result, we adapt a construction
due to Eppstein and Rystsov~\cite{Epp90,Rys80} which
uses the $\NP$-complete SAT problem~\cite{cook1971}.

\begin{quote}
    \textsc{SAT} \\ 
    \textrm{Input:} \textit{A set $X$ of $n$ boolean variables and a set $C$ of $m$ clauses;} \\ 
    \textrm{Question:} \textit{Does there exist an assignment of values to the variables in $X$
    such that all clauses in $C$ are satisfied?}
\end{quote}
First, we single out those constraint languages that give $\NP$-hard problems.

\begin{toappendix}
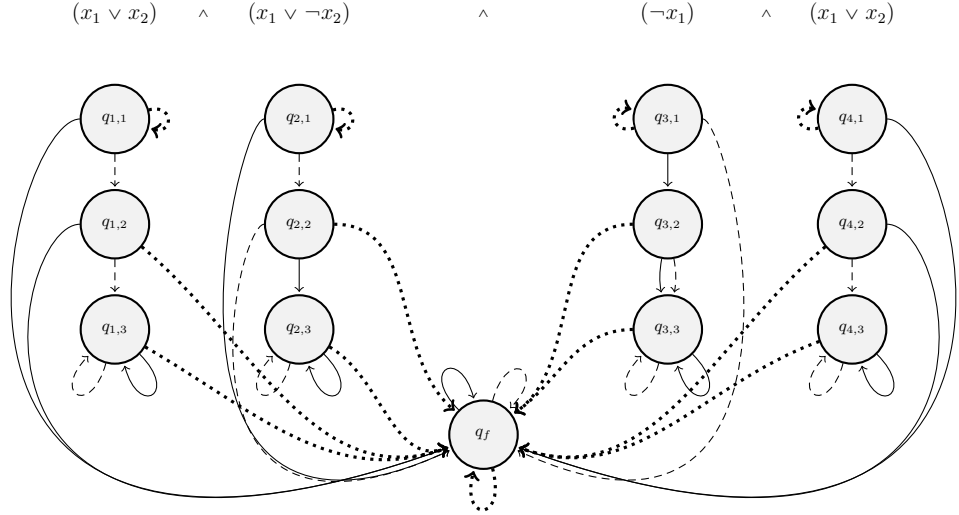
\begin{figure}
\centering
		    \begin{tikzpicture}[transform shape,scale=0.7,shorten >=1pt,node distance=2cm and 3.5cm,on grid,auto, state/.style={circle, draw, minimum size=1.3cm, fill=black!5, thick}] 
			
				\node (clause_1) {\large$(x_1 \vee x_2)$};
				\node[right=of clause_1] (clause_2) {\large $(x_1 \vee \neg x_2)$};
				\node[right=of clause_2] (clause_dots) {$\wedge$};
				\node[right=of clause_dots] (clause_{m-1}) {\large $(\neg x_1)$}; 
				\node[right=of clause_{m-1}] (clause_m) {\large $(x_1 \vee x_2)$};
				\node (label_wedge_1) at ($(clause_1)!0.48!(clause_2)$) {$\wedge$}; 
				\node (label_wedge_2) at ($(clause_2)!0.52!(clause_dots)$) {};
				\node (label_wedge_{m-2}) at ($(clause_dots)!0.46!(clause_{m-1})$) {};
				\node (label_wedge_{m-1}) at ($(clause_{m-1})!0.535!(clause_m)$) {$\wedge$};
				
				
				\node[below=of clause_dots] (q_{3ff,1}) {};
				\node[below=of q_{3ff,1}] (q_{3ff,2}) {};
				\node[below=of q_{3ff,2}] (q_{3ff,3}) {};

				\node[state, left=of q_{3ff,1}] (q_{2,1}) {$q_{2,1}$};
				\node[state, below=of q_{2,1}] (q_{2,2}) {$q_{2,2}$};
				\node[state, below=of q_{2,2}] (q_{2,3}) {$q_{2,3}$};

				\node[state, left=of q_{2,1}] (q_{1,1}) {$q_{1,1}$};
				\node[state, below=of q_{1,1}] (q_{1,2}) {$q_{1,2}$};
				\node[state, below=of q_{1,2}] (q_{1,3}) {$q_{1,3}$};

				\node[state, right=of q_{3ff,1}] (q_{m-1,1}) {$q_{3,1}$};
				\node[state, below=of q_{m-1,1}] (q_{m-1,2}) {$q_{3,2}$};
				\node[state, below=of q_{m-1,2}] (q_{m-1,3}) {$q_{3,3}$};

				\node[state, right=of q_{m-1,1}] (q_{m,1}) {$q_{4,1}$};
				\node[state, below=of q_{m,1}] (q_{m,2}) {$q_{4,2}$};
				\node[state, below=of q_{m,2}] (q_{m,3}) {$q_{4,3}$};

				\node[state, below=of q_{3ff,3}](q_t) {$q_f$};

				\path[use as bounding box] (0, 0) rectangle (0, 0); 

				\path[->, every loop/.style={looseness=3}]
				(q_{2,1})
					edge[out=0, in=190, controls=+(180:1.2cm) and +(205:7.5cm)] (q_t)
					edge[densely dashed] (q_{2,2})
					edge[loop right,very thick,dotted] ()
				(q_{2,2})
					edge[densely dashed, out=0, in=180, controls=+(180:1.2cm) and +(205:6.85cm)] (q_t)
					edge (q_{2,3})
					edge[very thick,dotted,out=0,in=140] (q_t)
				(q_{2,3}) edge[very thick,dotted,out=330,in=200] (q_t)
					;
				
				\path[->, every loop/.style={looseness=3}] 
				(q_{1,1})
					edge[out=0, in=240, controls=+(180:2cm) and +(200:12.5cm)] (q_t) 
					edge[densely dashed] (q_{1,2})
					edge[loop right,very thick,dotted] () 
				(q_{1,2})
					edge[out=0, in=230, controls=+(180:2cm) and +(200:11.25cm)] (q_t)
					edge[densely dashed] (q_{1,3})
				edge [very thick,dotted, out=320, in=200] (q_t)
				(q_{1,3}) edge[very thick,dotted,out=330,in=200] (q_t)
					;

				\path[->, every loop/.style={looseness=3}]
				(q_{m-1,1})
					edge[densely dashed, out=180, in=350, controls=+(0:1.2cm) and +(335:7.5cm)] (q_t)
					edge (q_{m-1,2})
					edge[-,draw=white,line width=7pt,loop left] node[line width = 3pt,draw=white,fill=white, text=white, ultra thin] {$\sigma_{k'}$} ()
					edge[loop left,very thick,dotted] ()
				(q_{m-1,2})
					edge[bend right=10] (q_{m-1,3})
					edge[densely dashed, bend left=10] (q_{m-1,3})
				    edge[very thick,dotted,out=180,in=35] (q_t)
				(q_{m-1,3}) edge[very thick,dotted,out=180,in=35] (q_t)
					;
				
				\path[->, every loop/.style={looseness=3}]		
				(q_{m,1})
					edge[out=180,in=300, controls=+(0:2cm) and +(340:12.5cm)] (q_t)
					edge[densely dashed] (q_{m,2})
					edge[-,draw=white,line width=7pt,loop left] node[line width = 3pt,draw=white,fill=white, text=white, ultra thin] {$\sigma_{k'}$} ()
					edge[loop left,very thick,dotted]  ()
				(q_{m,2})
					edge[out=180,in=310, controls=+(0:2cm) and +(340:11.25cm)] (q_t)
					edge[densely dashed] (q_{m,3})
					edge[very thick,dotted,out=220,in=340] (q_t)
				(q_{m,3}) edge[very thick,dotted,out=200,in=340] (q_t)
				;

				\path[->, every loop/.style={in=45, out=75}] (q_t) edge[densely dashed, loop above] ();
				\path[->, every loop/.style={in=105, out=135}] (q_t) edge[loop above] ();
				\path[->, every loop/.style={very thick,dotted}] (q_t) edge[loop below] ();
				
				\path[->, every loop/.style={in=225, out=255}] (q_{1,3}) edge[densely dashed, loop below] ();
				\path[->, every loop/.style={in=225, out=255}] (q_{2,3}) edge[densely dashed, loop below] ();
				\path[->, every loop/.style={in=225, out=255}] (q_{m-1,3}) edge[densely dashed, loop below] ();
				\path[->, every loop/.style={in=225, out=255}] (q_{m,3}) edge[densely dashed, loop below] ();
				
				\path[->, every loop/.style={in=285, out=315}] (q_{1,3}) edge[loop above] ();
				\path[->, every loop/.style={in=285, out=315}] (q_{2,3}) edge[loop above] ();
				\path[->, every loop/.style={in=285, out=315}] (q_{m-1,3}) edge[loop above] ();
				\path[->, every loop/.style={in=285, out=315}] (q_{m,3}) edge[loop above] ();
			\end{tikzpicture}
		\vspace{11pt} 
	\caption{Example for the SAT formula
	$(x_1 \lor x_2) \land (x_1 \lor \neg x_2)
	\land (\neg x_1) \land (x_1 \lor x_2)$
	of the reduction used in the proof
	of Proposition~\ref{prop:np_complete_cases}
	for the constraint language $a(b+c)^*$.
	The transitions for $a$ are drawn with thick dotted lines,
	the transitions labeled by~$b$ are drawn by dashed lines and
	those labeled by~$c$ are drawn by solid lines.}
	\label{fig:a_b_plus_c}	
\end{figure}
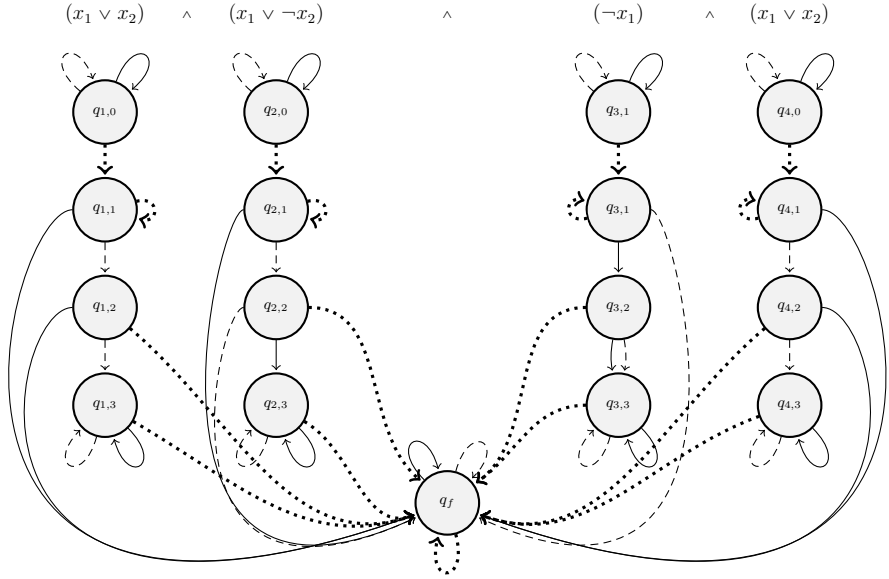
\begin{figure}
\centering
		    \begin{tikzpicture}[transform shape,scale=0.65,shorten >=1pt,node distance=2cm and 3.5cm,on grid,auto, state/.style={circle, draw, minimum size=1.3cm, fill=black!5, thick}] 
			
				\node (clause_1) {\large$(x_1 \vee x_2)$};
				\node[right=of clause_1] (clause_2) {\large $(x_1 \vee \neg x_2)$};
				\node[right=of clause_2] (clause_dots) {$\wedge$};
				\node[right=of clause_dots] (clause_{m-1}) {\large $(\neg x_1)$}; 
				\node[right=of clause_{m-1}] (clause_m) {\large $(x_1 \vee x_2)$};
				\node (label_wedge_1) at ($(clause_1)!0.48!(clause_2)$) {$\wedge$}; 
				\node (label_wedge_2) at ($(clause_2)!0.52!(clause_dots)$) {};
				\node (label_wedge_{m-2}) at ($(clause_dots)!0.46!(clause_{m-1})$) {};
				\node (label_wedge_{m-1}) at ($(clause_{m-1})!0.535!(clause_m)$) {$\wedge$};
				
				
				\node[below=of clause_dots] (q_{3ff,0}) {};
				\node[below=of q_{3ff,0}] (q_{3ff,1}) {};
				\node[below=of q_{3ff,1}] (q_{3ff,2}) {};
				\node[below=of q_{3ff,2}] (q_{3ff,3}) {};

				\node[state, left=of q_{3ff,0}] (q_{2,0}) {$q_{2,0}$};
				\node[state, below=of q_{2,0}] (q_{2,1}) {$q_{2,1}$};
				\node[state, below=of q_{2,1}] (q_{2,2}) {$q_{2,2}$};
				\node[state, below=of q_{2,2}] (q_{2,3}) {$q_{2,3}$};

				\node[state, left=of q_{2,0}] (q_{1,0}) {$q_{1,0}$};
				\node[state, left=of q_{2,1}] (q_{1,1}) {$q_{1,1}$};
				\node[state, below=of q_{1,1}] (q_{1,2}) {$q_{1,2}$};
				\node[state, below=of q_{1,2}] (q_{1,3}) {$q_{1,3}$};

				\node[state, right=of q_{3ff,0}] (q_{m-1,0}) {$q_{3,1}$};
				\node[state, below=of q_{m-1,0}] (q_{m-1,1}) {$q_{3,1}$};
				\node[state, below=of q_{m-1,1}] (q_{m-1,2}) {$q_{3,2}$};
				\node[state, below=of q_{m-1,2}] (q_{m-1,3}) {$q_{3,3}$};

				\node[state, right=of q_{m-1,0}] (q_{m,0}) {$q_{4,0}$};
				\node[state, below=of q_{m,0}] (q_{m,1}) {$q_{4,1}$};
				\node[state, below=of q_{m,1}] (q_{m,2}) {$q_{4,2}$};
				\node[state, below=of q_{m,2}] (q_{m,3}) {$q_{4,3}$};

				\node[state, below=of q_{3ff,3}](q_t) {$q_f$};

				\path[use as bounding box] (0, 0) rectangle (0, 0); 

				\path[->, every loop/.style={looseness=3}]
				(q_{2,0}) edge[very thick,dotted] (q_{2,1})
				(q_{1,0}) edge[very thick,dotted] (q_{1,1})
				(q_{m-1,0}) edge[very thick,dotted] (q_{m-1,1})
				(q_{m,0}) edge[very thick,dotted] (q_{m,1})
				;
				\path[->, every loop/.style={looseness=3}]
				(q_{2,1})
					edge[out=0, in=190, controls=+(180:1.2cm) and +(205:7.5cm)] (q_t)
					edge[densely dashed] (q_{2,2})
					edge[loop right,very thick,dotted] ()
				(q_{2,2})
					edge[densely dashed, out=0, in=180, controls=+(180:1.2cm) and +(205:6.85cm)] (q_t)
					edge (q_{2,3})
					edge[very thick,dotted,out=0,in=140] (q_t)
				(q_{2,3}) edge[very thick,dotted,out=330,in=200] (q_t)
					;
				
				\path[->, every loop/.style={looseness=3}] 
				(q_{1,1})
					edge[out=0, in=240, controls=+(180:2cm) and +(200:12.5cm)] (q_t) 
					edge[densely dashed] (q_{1,2})
					edge[loop right,very thick,dotted] () 
				(q_{1,2})
					edge[out=0, in=230, controls=+(180:2cm) and +(200:11.25cm)] (q_t)
					edge[densely dashed] (q_{1,3})
				edge [very thick,dotted, out=320, in=200] (q_t)
				(q_{1,3}) edge[very thick,dotted,out=330,in=200] (q_t)
					;

				\path[->, every loop/.style={looseness=3}]
				(q_{m-1,1})
					edge[densely dashed, out=180, in=350, controls=+(0:1.2cm) and +(335:7.5cm)] (q_t)
					edge (q_{m-1,2})
					edge[-,draw=white,line width=7pt,loop left] node[line width = 3pt,draw=white,fill=white, text=white, ultra thin] {$\sigma_{k'}$} ()
					edge[loop left,very thick,dotted] ()
				(q_{m-1,2})
					edge[bend right=10] (q_{m-1,3})
					edge[densely dashed, bend left=10] (q_{m-1,3})
				    edge[very thick,dotted,out=180,in=35] (q_t)
				(q_{m-1,3}) edge[very thick,dotted,out=180,in=35] (q_t)
					;
				
				\path[->, every loop/.style={looseness=3}]		
				(q_{m,1})
					edge[out=180,in=300, controls=+(0:2cm) and +(340:12.5cm)] (q_t)
					edge[densely dashed] (q_{m,2})
					edge[-,draw=white,line width=7pt,loop left] node[line width = 3pt,draw=white,fill=white, text=white, ultra thin] {$\sigma_{k'}$} ()
					edge[loop left,very thick,dotted]  ()
				(q_{m,2})
					edge[out=180,in=310, controls=+(0:2cm) and +(340:11.25cm)] (q_t)
					edge[densely dashed] (q_{m,3})
					edge[very thick,dotted,out=220,in=340] (q_t)
				(q_{m,3}) edge[very thick,dotted,out=200,in=340] (q_t)
				;

				\path[->, every loop/.style={in=45, out=75}] (q_t) edge[densely dashed, loop above] ();
				\path[->, every loop/.style={in=105, out=135}] (q_t) edge[loop above] ();
				\path[->, every loop/.style={very thick,dotted}] (q_t) edge[loop below] ();
				
				\path[->, every loop/.style={in=225, out=255}] (q_{1,3}) edge[densely dashed, loop below] ();
				\path[->, every loop/.style={in=225, out=255}] (q_{2,3}) edge[densely dashed, loop below] ();
				\path[->, every loop/.style={in=225, out=255}] (q_{m-1,3}) edge[densely dashed, loop below] ();
				\path[->, every loop/.style={in=225, out=255}] (q_{m,3}) edge[densely dashed, loop below] ();
				
				\path[->, every loop/.style={in=285, out=315}] (q_{1,3}) edge[loop above] ();
				\path[->, every loop/.style={in=285, out=315}] (q_{2,3}) edge[loop above] ();
				\path[->, every loop/.style={in=285, out=315}] (q_{m-1,3}) edge[loop above] ();
				\path[->, every loop/.style={in=285, out=315}] (q_{m,3}) edge[loop above] ();
				
			    \path[->, every loop/.style={in=110, out=140}] (q_{1,0}) edge[densely dashed, loop below] ();
				\path[->, every loop/.style={in=110, out=140}] (q_{2,0}) edge[densely dashed, loop below] ();
				\path[->, every loop/.style={in=110, out=140}] (q_{m-1,0}) edge[densely dashed, loop below] ();
				\path[->, every loop/.style={in=110, out=140}] (q_{m,0}) edge[densely dashed, loop below] ();
				
				\path[->, every loop/.style={in=40, out=70}] (q_{1,0}) edge[loop above] ();
				\path[->, every loop/.style={in=40, out=70}] (q_{2,0}) edge[loop above] ();
				\path[->, every loop/.style={in=40, out=70}] (q_{m-1,0}) edge[loop above] ();
				\path[->, every loop/.style={in=40, out=70}] (q_{m,0}) edge[loop above] ();
			\end{tikzpicture}
		\vspace{11pt} 
	\caption{A concrete example, for the SAT formula
	$(x_1 \lor x_2) \land (x_1 \lor \neg x_2)
	\land (\neg x_1) \land (x_1 \lor x_2)$,
	of the reduction used in the proof
	of Proposition~\ref{prop:np_complete_cases}
	for $(a+b)^*c(a+b)^*$.
	The transitions for the letter $c$ are drawn with dotted lines,
	the transitions labeled by~$a$ are drawn by dashed lines and
	those labeled by~$b$ are drawn by solid lines.}
	\label{fig:extended_Eppstein}	
\end{figure}
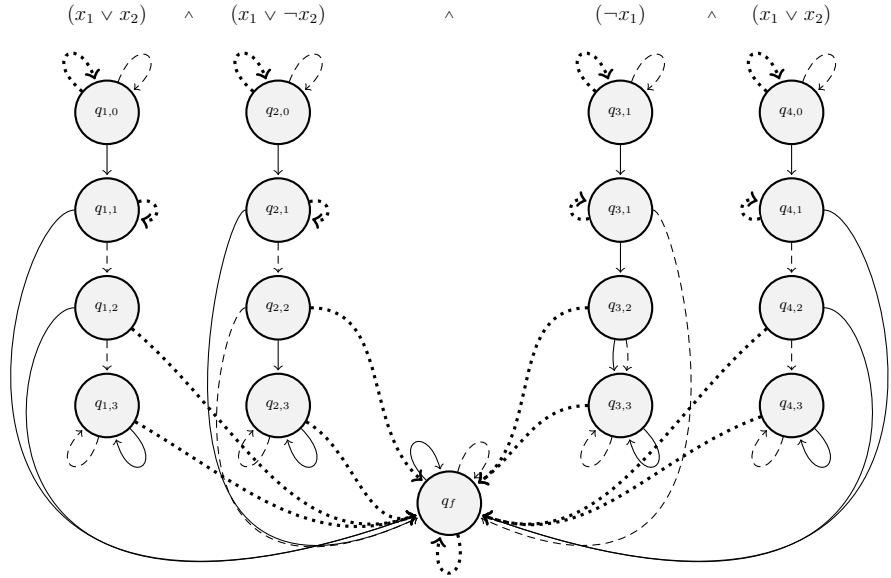
\begin{figure}
\centering
		    \begin{tikzpicture}[transform shape,scale=0.65,shorten >=1pt,node distance=2cm and 3.5cm,on grid,auto, state/.style={circle, draw, minimum size=1.3cm, fill=black!5, thick}] 
			
				\node (clause_1) {\large$(x_1 \vee x_2)$};
				\node[right=of clause_1] (clause_2) {\large $(x_1 \vee \neg x_2)$};
				\node[right=of clause_2] (clause_dots) {$\wedge$};
				\node[right=of clause_dots] (clause_{m-1}) {\large $(\neg x_1)$}; 
				\node[right=of clause_{m-1}] (clause_m) {\large $(x_1 \vee x_2)$};
				\node (label_wedge_1) at ($(clause_1)!0.48!(clause_2)$) {$\wedge$}; 
				\node (label_wedge_2) at ($(clause_2)!0.52!(clause_dots)$) {};
				\node (label_wedge_{m-2}) at ($(clause_dots)!0.46!(clause_{m-1})$) {};
				\node (label_wedge_{m-1}) at ($(clause_{m-1})!0.535!(clause_m)$) {$\wedge$};
				
				
				\node[below=of clause_dots] (q_{3ff,0}) {};
				\node[below=of q_{3ff,0}] (q_{3ff,1}) {};
				\node[below=of q_{3ff,1}] (q_{3ff,2}) {};
				\node[below=of q_{3ff,2}] (q_{3ff,3}) {};

				\node[state, left=of q_{3ff,0}] (q_{2,0}) {$q_{2,0}$};
				\node[state, below=of q_{2,0}] (q_{2,1}) {$q_{2,1}$};
				\node[state, below=of q_{2,1}] (q_{2,2}) {$q_{2,2}$};
				\node[state, below=of q_{2,2}] (q_{2,3}) {$q_{2,3}$};

				\node[state, left=of q_{2,0}] (q_{1,0}) {$q_{1,0}$};
				\node[state, left=of q_{2,1}] (q_{1,1}) {$q_{1,1}$};
				\node[state, below=of q_{1,1}] (q_{1,2}) {$q_{1,2}$};
				\node[state, below=of q_{1,2}] (q_{1,3}) {$q_{1,3}$};

				\node[state, right=of q_{3ff,0}] (q_{m-1,0}) {$q_{3,1}$};
				\node[state, below=of q_{m-1,0}] (q_{m-1,1}) {$q_{3,1}$};
				\node[state, below=of q_{m-1,1}] (q_{m-1,2}) {$q_{3,2}$};
				\node[state, below=of q_{m-1,2}] (q_{m-1,3}) {$q_{3,3}$};

				\node[state, right=of q_{m-1,0}] (q_{m,0}) {$q_{4,0}$};
				\node[state, below=of q_{m,0}] (q_{m,1}) {$q_{4,1}$};
				\node[state, below=of q_{m,1}] (q_{m,2}) {$q_{4,2}$};
				\node[state, below=of q_{m,2}] (q_{m,3}) {$q_{4,3}$};

				\node[state, below=of q_{3ff,3}](q_t) {$q_f$};

				\path[use as bounding box] (0, 0) rectangle (0, 0); 

				\path[->, every loop/.style={looseness=3}]
				(q_{2,0}) edge (q_{2,1})
				(q_{1,0}) edge (q_{1,1})
				(q_{m-1,0}) edge (q_{m-1,1})
				(q_{m,0}) edge (q_{m,1})
				;
				\path[->, every loop/.style={looseness=3}]
				(q_{2,1})
					edge[out=0, in=190, controls=+(180:1.2cm) and +(205:7.5cm)] (q_t)
					edge[densely dashed] (q_{2,2})
					edge[loop right,very thick,dotted] ()
				(q_{2,2})
					edge[densely dashed, out=0, in=180, controls=+(180:1.2cm) and +(205:6.85cm)] (q_t)
					edge (q_{2,3})
					edge[very thick,dotted,out=0,in=140] (q_t)
				(q_{2,3}) edge[very thick,dotted,out=330,in=200] (q_t)
					;
				
				\path[->, every loop/.style={looseness=3}] 
				(q_{1,1})
					edge[out=0, in=240, controls=+(180:2cm) and +(200:12.5cm)] (q_t) 
					edge[densely dashed] (q_{1,2})
					edge[loop right,very thick,dotted] () 
				(q_{1,2})
					edge[out=0, in=230, controls=+(180:2cm) and +(200:11.25cm)] (q_t)
					edge[densely dashed] (q_{1,3})
				edge [very thick,dotted, out=320, in=200] (q_t)
				(q_{1,3}) edge[very thick,dotted,out=330,in=200] (q_t)
					;

				\path[->, every loop/.style={looseness=3}]
				(q_{m-1,1})
					edge[densely dashed, out=180, in=350, controls=+(0:1.2cm) and +(335:7.5cm)] (q_t)
					edge (q_{m-1,2})
					edge[-,draw=white,line width=7pt,loop left] node[line width = 3pt,draw=white,fill=white, text=white, ultra thin] {$\sigma_{k'}$} ()
					edge[loop left,very thick,dotted] ()
				(q_{m-1,2})
					edge[bend right=10] (q_{m-1,3})
					edge[densely dashed, bend left=10] (q_{m-1,3})
				    edge[very thick,dotted,out=180,in=35] (q_t)
				(q_{m-1,3}) edge[very thick,dotted,out=180,in=35] (q_t)
					;
				
				\path[->, every loop/.style={looseness=3}]		
				(q_{m,1})
					edge[out=180,in=300, controls=+(0:2cm) and +(340:12.5cm)] (q_t)
					edge[densely dashed] (q_{m,2})
					edge[-,draw=white,line width=7pt,loop left] node[line width = 3pt,draw=white,fill=white, text=white, ultra thin] {$\sigma_{k'}$} ()
				    edge[loop left,very thick,dotted] ()
				(q_{m,2})
					edge[out=180,in=310, controls=+(0:2cm) and +(340:11.25cm)] (q_t)
					edge[densely dashed] (q_{m,3})
					edge[very thick,dotted,out=220,in=340] (q_t)
				(q_{m,3}) edge[very thick,dotted,out=200,in=340] (q_t)
				;

				\path[->, every loop/.style={in=45, out=75}] (q_t) edge[densely dashed, loop above] ();
				\path[->, every loop/.style={in=105, out=135}] (q_t) edge[loop above] ();
				\path[->, every loop/.style={very thick,dotted}] (q_t) edge[loop below] ();
				
				\path[->, every loop/.style={in=225, out=255}] (q_{1,3}) edge[densely dashed, loop below] ();
				\path[->, every loop/.style={in=225, out=255}] (q_{2,3}) edge[densely dashed, loop below] ();
				\path[->, every loop/.style={in=225, out=255}] (q_{m-1,3}) edge[densely dashed, loop below] ();
				\path[->, every loop/.style={in=225, out=255}] (q_{m,3}) edge[densely dashed, loop below] ();
				
				\path[->, every loop/.style={in=285, out=315}] (q_{1,3}) edge[loop above] ();
				\path[->, every loop/.style={in=285, out=315}] (q_{2,3}) edge[loop above] ();
				\path[->, every loop/.style={in=285, out=315}] (q_{m-1,3}) edge[loop above] ();
				\path[->, every loop/.style={in=285, out=315}] (q_{m,3}) edge[loop above] ();
				
			    \path[->, every loop/.style={in=110, out=140}] (q_{1,0}) edge[very thick,dotted, loop below] ();
				\path[->, every loop/.style={in=110, out=140}] (q_{2,0}) edge[very thick,dotted, loop below] ();
				\path[->, every loop/.style={in=110, out=140}] (q_{m-1,0}) edge[very thick,dotted, loop below] ();
				\path[->, every loop/.style={in=110, out=140}] (q_{m,0}) edge[very thick,dotted, loop below] ();
				
				\path[->, every loop/.style={in=40, out=70}] (q_{1,0}) edge[densely dashed,loop above] ();
				\path[->, every loop/.style={in=40, out=70}] (q_{2,0}) edge[densely dashed,loop above] ();
				\path[->, every loop/.style={in=40, out=70}] (q_{m-1,0}) edge[densely dashed,loop above] ();
				\path[->, every loop/.style={in=40, out=70}] (q_{m,0}) edge[densely dashed,loop above] ();
			\end{tikzpicture}
		\vspace{11pt} 
	\caption{Example for 
	$(x_1 \lor x_2) \land (x_1 \lor \neg x_2)
	\land (\neg x_1) \land (x_1 \lor x_2)$
	of the reduction used in the proof
	of Proposition~\ref{prop:np_complete_cases}
	for $(a+b)^*c(b+c)^*$.
	The $a$-transitions are drawn with dotted lines,
	the $b$-transitions by solid lines and
	the $c$-transitions by dashed lines. The self-loops at the $q_{i,1}$'s could also be changed to transitions ending in $q_f$.}
	\label{fig:a_plus_b_c_b_plus_c}	
\end{figure}
\end{toappendix}

\begin{propositionrep}
\label{prop:np_complete_cases}
 For the following
 constraint languages, 
 the constrained synchronization problem
 for weakly acyclic automata is $\NP$-hard:
 $$
  \begin{array}{llll}
    a(b+c)^*        & (a+b+c)(a+b)^*  & (a+b)(a+c)^*  \\
    (a+b)^*c(a+b)^* & a^*b(a+c)^*     &  a^*(b+c)(a+b)^* \\
    a^*b(b+c)^*     & (a+b)^*c(b+c)^* & a^*(b+c)(b+c)^*.
  \end{array}
 $$ 
\end{propositionrep}
\begin{proofsketch}
 We only sketch the case $L(\mathcal B) = (a+b)^*c(b+c)^*$, the other cases could be handled similarly.
 We adapt a reduction
 by Eppstein and Rystsov~\cite{Epp90,Rys80}
 to show $\NP$-hardness for the decision variant of the problem of a shortest synchronizing word.
 Given a SAT instance with variables $X = \{x_1, \ldots, x_n\}$
 and clauses $C = \{c_1, \ldots, c_m\}$,
 we construct a weakly acyclic automaton $\mathcal A = (\Sigma, Q, \delta)$ over the alphabet $\{a,b,c\}$
 with states $q_{i,j}$ for $1 \le i \le m$ and $0 \le j \le n + 1$, plus
 a sink state $q_f$.
 Then $\delta$ is defined, for $i \in \{1,\ldots,m\}$ and $j \in \{1,\ldots,n\}$, as
 \[
  \delta(q_{i,j}, b) = \left\{ 
  \begin{array}{ll}
   q_{i, j + 1} & \mbox{if }  \overline x_j \in c_i \lor \{x_j, \overline x_j\}\cap c_i = \emptyset; \\ 
   q_f          & \mbox{if } x_j \in c_i;
  \end{array}
  \right.
 \] 
 and, symmetrically,
 \[
   \delta(q_{i,j}, c) = \left\{ 
  \begin{array}{ll}
   q_{i, j + 1} & \mbox{if } x_j \in c_i \lor \{x_j, \overline x_j\}\cap c_i = \emptyset; \\ 
   q_f          & \mbox{if } \overline x_j \in c_i.
  \end{array}
  \right.
 \] 
 Furthermore, for $i \in \{1,\ldots,m\}$ and $j \in \{0,\ldots,n+1\}$,
 \[
  \delta(q_{i, j}, a) = \left\{ 
  \begin{array}{ll}
   q_{i,j} & \mbox{if } j \in \{0,1\}; \\ 
   q_f     & \mbox{if } j \notin \{0,1\}.
  \end{array}
  \right.
 \]
 Lastly, for $i \in \{1,\ldots,m\}$, we set
 $\delta(q_{i, n+1}, b) = \delta(q_{i, n+1}, c) = q_{i,n+1}$, $\delta(q_{i, 0}, b) =  q_{i,0}$,
 $\delta(q_{i, 0}, c) = q_{i,1}$
 and $q_f = \delta(q_f, a) = \delta(q_f, b) = \delta(q_f, c)$.
 Note that we have $\{ q_{1,1} , \ldots, q_{m,1} \} \subseteq \delta(Q, uc)$
 for any $u \in \{a,b\}^*$ and, for $v \in \{b,c\}^*$, $\delta(q_{i,1}, v) = q_f$
 if and only if some symbol in $v$ at a position smaller or equal than $n$
 branches out of the strand $q_{i,1}, \cdots, q_{i,n}$, which means
 $v$ could be identified with a satisfying assignment for the clause $c_i$.
 Conversely, if we have a satisfying assignment, construct
 a word $v = v_1 \cdots v_n \in \{b,c\}^*$ by setting $v_i = b$
 if the $i$-th variable is set to one, and $v_i = c$ otherwise.
 Then, $\delta(Q, acv) =  \{q_f\}$. So, we can show that $\mathcal A$ has a synchronizing
 word in $L(\mathcal B)$ if and only
 if there exists a satisfying assignment for all
 clauses in $C$.~\qed
\end{proofsketch}
\begin{proof}
 The following reductions are all adaptions of a reduction from Eppstein and Rystsov~\cite{Epp90,Rys80}
 to show $\NP$-hardness for the decision variant of the problem of a shortest synchronizing word.
 The reduction works by constructing for every clause a linear path.
 The states on these paths correspond to the variables, and we can only leave this
 path and end in a synchronizing sink state at those states whose corresponding literals
 are contained in the clause associated to the path.
 The additional letter is used, in some sense, to reset all states either to the start states of the paths
 or to the sink state.

 \begin{enumerate}
 \item \label{case:a_b_plus_c}
  The constraint language $a(b+c)^*$.
  The adaption is similar to a reduction used in~\cite[Theorem 4]{Ryzhikov19a}.

  \medskip 
  
  \underline{Construction:} Given a SAT instance with variables $X = \{x_1, \ldots, x_n\}$
 and clauses $C = \{c_1, \ldots, c_m\}$,
 we construct a weakly acyclic automaton $\mathcal A = (\Sigma, Q, \delta)$ over the alphabet $\{a,b,c\}$
 with states $q_{i,j}$ for $1 \le i \le m$ and $1 \le j \le n + 1$, plus
 a sink state $q_f$.
 The transition function $\delta$ is defined, for $i \in \{1,\ldots,m\}$
 and $j \in \{1,\ldots,n\}$,
 as
 \[
  \delta(q_{i,j}, b) = \left\{ 
  \begin{array}{ll}
   q_{i, j + 1} & \mbox{if } x_j \in c_i \lor \{x_j, \overline x_j\}\cap c_i = \emptyset; \\ 
   q_f          & \mbox{if } \overline x_j \in c_i;
  \end{array}
  \right.
 \] 
 and, symmetrically,
 \[
   \delta(q_{i,j}, c) = \left\{ 
  \begin{array}{ll}
   q_{i, j + 1} & \mbox{if } \overline x_j \in c_i \lor \{x_j, \overline x_j\}\cap c_i = \emptyset; \\ 
   q_f          & \mbox{if } x_j \in c_i.
  \end{array}
  \right.
 \] 
 Furthermore, for $i \in \{1,\ldots,m\}$ and $j \in \{1,\ldots,n+1\}$,
 we set
 \[
  \delta(q_{i, j}, a) = \left\{ 
  \begin{array}{ll}
   q_{i,1} & \mbox{if } j = 1; \\ 
   q_f     & \mbox{if } j \ne 0.
  \end{array}
  \right.
 \]
 Lastly, for $i \in \{1,\ldots,m\}$, we set 
 $\delta(q_{i, n+1}, b) = \delta(q_{i, n+1}, c) = q_{i,n+1}$,
 and $q_f = \delta(q_f, a) = \delta(q_f, b) = \delta(q_f, c)$.
 See Figure~\ref{fig:a_b_plus_c} for an example of the reduction.
 
 \medskip

 \underline{Verification:} Suppose $\mathcal A$ has a synchronizing word $w \in L(B)$.
 Then, $w = au$ with $u \in \{b,c\}^*$.
 As $q_f$ is a sink state, we have $\delta(Q, w) = \{q_f\}$.
 Consider a state $q_{i,1}$. As $a$ labels
 a self-loop on this state, we have $\delta(q_{i,1}, u) = s_f$.
 However,  $\delta(q_{i,1}, u) = s_f$ holds precisely if $u$
 has a prefix $v\in \{b,c\}^n$
 with $v = v_1 \cdots v_n$ for $\{v_1, \ldots, v_n\}\subseteq \Sigma$
 such that for any clause $c_i$ there exists $r \in \{1,\ldots, n\}$
 with $v_r = b$ if $\overline x_r \in c_i$
 or $v_r = c$ if $x_r \in c_i$.
 Hence, we get a satisfying assignment for all clauses
 by setting $x_j = 0$ if $v_j = b$
 and $x_j = 1$ if $v_j = c$ for $j \in \{1,\ldots,n\}$.

 Conversely, suppose  a satisfying assignment
 for the variables $x_1, \ldots, x_n$ exists.
 Then, set $u = u_1 \cdots u_n$
 with $u_j = b$ if $x_j = 0$
 and $u_j = c$ if $x_j = 1$ for $j \in \{1,\ldots,n\}$.
 Then, we have $\delta(q_{i,1}, u) = s_f$
 for $i \in \{1,\ldots,m\}$.
 As $\delta(Q, a) = \{ q_{1,1}, \ldots, q_{m,1}, q_f\}$,
 we find $\delta(Q, au) = \{s_f\}$.

 \medskip

 \item The constraint language $(a+b+c)(a+b)^*$.
 \label{case:a_b_c_a_plus_b}
  We can use the same reduction as 
  in Case~\ref{case:a_b_plus_c}, but
  with the letters changed: the letter $c$
  acts like the letter $a$ in Case~\ref{case:a_b_plus_c}
  and the letter $a$ like the letter $c$ before.
  Note that we can use the letter $c$ if we have a satisfying assignment
  for the given SAT formula, but not conversely. However, by only 
  investigating the paths taken from the start states $q_{i,1}$,
  we can read of a satisfying assignment.
  
 \item The constraint language $(a+b)(a+c)^*$.
  We can use the same reduction as 
  in Case~\ref{case:a_b_plus_c}, but
  with the letters changed: the letter $b$
  acts like the letter $a$ in Case~\ref{case:a_b_plus_c}
  and the letter $a$ like the letter $b$ before.
  Note that we can use the letter $b$ if we have a satisfying assignment
  for the given SAT formula, but we do not need to use it in in the other
  direction of the reduction. However, by only 
  investigating the paths taken from the start states $q_{i,1}$,
  we can read of a satisfying assignment.
 
 \item \label{case:a_plus_b_c_a_plus_b} The constraint language $(a+b)^*c(a+b)^*$.
  %
  %
  %
  Here, we need a different construction.
  The reason is that the special letter that is only
  allowed to be read once, which was the letter $a$
  in Case~\ref{case:a_b_plus_c} and is the letter $c$ here, had, in the previous cases,
  the property
  that before it, it was not possible to read
  anything, i.e., if it appeared in a word of the constraint
  language, it 
  appeared as the first letter of that word.
  Without this property, for example
  in the case we are considering now, we can read in
  any word long enough to drive everything into 
  the last states of the strands corresponding to the clauses,
  i.e., the states $q_{i,n+1}$ in the reduction of Case~\ref{case:a_b_plus_c}.
  Then, from this state we could read the special letter to map them to the
  the sink state.

  We can circumvent this by introducing additional state $q_{i,0}$.
  More formally, we give the complete construction next.
  
  \medskip

  \underline{Construction:} Given a SAT instance with variables $X = \{x_1, \ldots, x_n\}$
 and clauses $C = \{c_1, \ldots, c_m\}$,
 we construct a weakly acyclic automaton $\mathcal A = (\Sigma, Q, \delta)$ over the alphabet $\{a,b,c\}$
 with states $q_{i,j}$ for $1 \le i \le m$ and $0 \le j \le n + 1$, plus
 a sink state $q_f$.
 The transition function $\delta$ is defined 
 as, for $i \in \{1,\ldots,m\}$ and $j \in \{1,\ldots,n\}$,
 \[
  \delta(q_{i,j}, a) = \left\{ 
  \begin{array}{ll}
   q_{i, j + 1} & \mbox{if } x_j \in c_i \lor \{x_j, \overline x_j\}\cap c_i = \emptyset; \\ 
   q_f          & \mbox{if } \overline x_j \in c_i;
  \end{array}
  \right.
 \] 
 and, symmetrically,
 \[
   \delta(q_{i,j}, b) = \left\{ 
  \begin{array}{ll}
   q_{i, j + 1} & \mbox{if } \overline x_j \in c_i \lor \{x_j, \overline x_j\}\cap c_i = \emptyset; \\ 
   q_f          & \mbox{if } x_j \in c_i.
  \end{array}
  \right.
 \] 
 Furthermore, for $i \in \{1,\ldots,m\}$ and $j \in \{1,\ldots,n+1\}$,
 \[
  \delta(q_{i, j}, c) = \left\{ 
  \begin{array}{ll}
   q_{i,1} & \mbox{if } j = 1; \\ 
   q_f     & \mbox{if } j \ne 0.
  \end{array}
  \right.
 \]
 Lastly, for $i \in \{1,\ldots,m\}$, we set
 $\delta(q_{i, n+1}, a) = \delta(q_{i, n+1}, b) = q_{i,n+1}$
 and
 \[
  \delta(q_{i, 0}, a) = \delta(q_{i,0}, b) = q_{i,0}, \quad
  \delta(q_{i, 0}, c) = q_{i,1}
 \] 
 and $q_f = \delta(q_f, a) = \delta(q_f, b) = \delta(q_f, c)$.
 
 \medskip 
 
 \underline{Verification:} Suppose $\mathcal A$ has a synchronizing word $w \in L(B)$.
 As $q_f$ is a sink state, we have $\delta(Q, w) = \{q_f\}$.
 Write $w = ucv$ with $u \in \{a,b\}^*$ and $v \in \{a,b\}^*$.
 For $i \in \{1,\ldots,m\}$, we have $\delta(q_{i,0}, uc) = q_{i,1}$.
 Hence, we must have $\delta(q_{i,1}, v) = s_f$.
 By construction, this is the case precisely
 if $v$ has a prefix $v' \in \{a,b\}^n$
 with $v' = v'_1 \cdots v'_n$ for $\{v'_1, \ldots, v'_n\}\subseteq \Sigma$
 such that for any clause $c_i$ there exists $r \in \{1,\ldots, n\}$
 with $v'_r = a$ if $\overline x_r \in c_i$
 or $v'_r = b$ if $x_r \in c_i$.
 Hence, we get a satisfying assignment for all clauses
 by setting $x_j = 0$ if $v'_j = a$
 and $x_j = 1$ if $v'_j = b$ for all $j \in \{1,\ldots,n\}$.

 Conversely, suppose  a satisfying assignment
 for the variables $x_1, \ldots, x_n$ exists.
 Then, set $u = u_1 \cdots u_n$
 with $u_j = a$ if $x_j = 0$
 and $u_j = b$ if $x_j = 1$ for $j \in \{1,\ldots,n\}$.
 Then, $\delta(q_{i,1}, u) = s_f$ for any $i \in \{1,\ldots,n\}$.
 Furthermore, we have $\delta(Q, c) = \{ q_{1,1}, \ldots, q_{m,1}, q_f \}$.
 Hence, $\delta(Q, cu) = \{s_f\}$.

  \medskip

 \item The constraint language $a^*b(a+c)^*$. We can use the same
 reduction as in Case~\ref{case:a_plus_b_c_a_plus_b}, but with
 the letter $b$ acting like the letter $c$ in Case~\ref{case:a_plus_b_c_a_plus_b}
 and the letter $c$ like the letter $b$ before.
 
 \item The constraint language $a^*(b+c)(a+b)^*$. Here, we can use the same
 reduction as in Case~\ref{case:a_plus_b_c_a_plus_b}. Note
 that the construction enforces that we have to use a word that uses
 the letter $c$.
 
 \item The constraint language $a^*b(b+c)^*$.
  We can use the same
 reduction as in Case~\ref{case:a_plus_b_c_a_plus_b}, but with
 the letter $b$ acting like the letter $c$ in Case~\ref{case:a_plus_b_c_a_plus_b}
 and the letter $c$ like the letter $b$ before.
 
 \item The constraint language $(a+b)^*c(b+c)^*$.  Here, the letter $a$
  will be the letter that is used as a special letter. However, we can not use the previous
  constructions, but use another one. The construction was already given in the main 
  text. Also, see Figure~\ref{fig:a_plus_b_c_b_plus_c}
  for an example of the reduction.
  
  \medskip
  
 \underline{Verification:} Suppose $\mathcal A$ has a synchronizing word $w \in L(B)$.
 As $q_f$ is a sink state, we have $\delta(Q, w) = \{q_f\}$.
 Write $w = ucv$ with $u \in \{a,b\}^*$ and $v \in \{b,c\}^*$.
 For $i \in \{1,\ldots,m\}$, we have $\delta(q_{i,0}, uc) = q_{i,1}$.
 Hence, we must have $\delta(q_{i,1}, v) = s_f$.
 By construction, this is the case precisely
 if $v$ has a prefix $v' \in \{b,c\}^{n'}$, $n' \le n$,
 with $v' = v'_1 \cdots v'_{n'}$ for $\{v'_1, \ldots, v'_{n'}\}\subseteq \Sigma$
 such that for any clause $c_i$ there exists $r \in \{1,\ldots, n\}$
 with $v'_r = b$ if $ x_r \in c_i$
 or $v'_r = c$ if $\overline x_r \in c_i$.
 Hence, we get a satisfying assignment for all clauses
 by setting $x_j = 1$ if $v'_j = b$
 and $x_j = 0$ if $v'_j = c$ for all $j \in \{1,\ldots,n'\}$,
 and setting the remaining variables to arbitrary values.

 Conversely, suppose  a satisfying assignment
 for the variables $x_1, \ldots, x_n$ exists.
 Then, set $u = u_1 \cdots u_n$
 with $u_j = b$ if $x_j = 1$
 and $u_j = c$ if $x_j = 0$ for $j \in \{1,\ldots,n\}$.
 Then, $\delta(q_{i,1}, u) = s_f$ for any $i \in \{1,\ldots,n\}$.
 Furthermore, we have $\delta(Q, ac) = \{ q_{1,1}, \ldots, q_{m,1}, q_f \}$.
 Hence, $\delta(Q, acu) = \{s_f\}$.
 
 \medskip 
 
 \item The constraint language $a^*(b+c)(b+c)^*$.
   We can use the same reduction as 
  in Case~\ref{case:a_b_plus_c}.
     
 \end{enumerate}
 So, we have handled all constraint languages from the statement
 and the proof is done.~\qed
\end{proof}

In the next two propositions, we handle those cases from the list given
in Theorem~\ref{thm:classification_MFCS_paper} that do not appear in Proposition~\ref{prop:np_complete_cases}.
It will turn out that for these cases, the complexity drops from $\PSPACE$-completeness
to polynomial time solvable.

\begin{propositionrep}
\label{prop:a_plus_b_c}
We have $((a+b)^*c)\textsc{-WAA-Constr-Sync} \in \PTIME$
\end{propositionrep}
\begin{proofsketch}
 By Lemma~\ref{lem:sync_waa_has_sink_state}, if $\mathcal A$ is synchronizable, it must possess
 a unique synchronizing sink state $s_f$. In that case, set $T = \delta^{-1}(s_f, c)$.
 Then, we have a synchronizing word in $(a+b)^*c$
 if and only if there exists a word $w \in (a+b)^*$
 such that $\delta_{|\{a,b\}}(Q, w) \subseteq T$ in $\mathcal A_{|\{a,b\}} = (\Sigma, Q, \delta_{|\{a,b\}})$.
 The latter problem is in~\PTIME\ by Theorem~\ref{thm:sync_into_subset_waa_in_P}.
 ~\qed
\end{proofsketch}
\begin{proof}
 By Lemma~\ref{lem:sync_waa_has_sink_state}, if $\mathcal A$ is synchronizable, it must possess
 a unique synchronizing sink state $s_f$, which is easily testable.
 So, if it does not possess a unique sink state, it could not have
 a synchronizing word in $(a+b)^*c$.
 If it has such a unique sink state $s_f$, set $T = \delta^{-1}(s_f, c)$.
 Then, we have a synchronizing word in $(a+b)^*c$
 if and only if there exists a word $w \in (a+b)^*$
 such that $\delta_{|\{a,b\}}(Q, w) \subseteq T$ in $\mathcal A_{|\{a,b\}} = (\Sigma, Q, \delta_{|\{a,b\}})$.
 The latter problem is solvable in polynomial time by Theorem~\ref{thm:sync_into_subset_waa_in_P}
 and we have a polynomial time procedure for our original problem.~\qed
\end{proof}

\begin{propositionrep}
\label{prop:a_plus_b_c_and_more}
 We have $((a+b)^*ca^*)\textsc{-WAA-Constr-Sync} \in \PTIME$
 and $((a+b)^*cc^*)\textsc{-WAA-Constr-Sync} \in \PTIME$.
\end{propositionrep}
\begin{proofsketch}
 By Lemma~\ref{lem:sync_waa_has_sink_state}, the automaton $\mathcal A$
 could only be synchronizing if it has a unique sink state $s_f$.
 In this case, set $S_i = \delta^{-1}(s_f, a^i)$ and $n = |Q|$.
 We have $S_i = S_n$ for any $i \ge n$.
 Then, for each $i \in \{0, \ldots, n\}$, set $T_i = \delta^{-1}(S_i, c)$
 and decide, which could be done in polynomial time by Theorem~\ref{thm:sync_into_subset_waa_in_P},
 if there exists a word $w \in \{a,b\}^*$ in $A_{|\{a,b\}} = (\{a,b\}, Q, \delta_{|\{a,b\}})$
 such that $\delta_{|\{a,b\}}(Q, w) \subseteq T_i$, which
 is equivalent to $\delta(Q, wca^i) = \{s_f\}$.~\qed
\end{proofsketch}
\begin{proof} 
 By Lemma~\ref{lem:sync_waa_has_sink_state}, the automaton $\mathcal A$
 could only be synchronizing if it has a unique sink state $s_f$.
 Consider the sets $S_i = \delta^{-1}(s_f, a^i)$.
 As every path in $\mathcal A$ has length at most $|Q|$ and the only loops are self-loops,
 we have 
 \[ 
  \{ S_i : i \ge 0 \} = \{ S_i : i \in \{0, \ldots, |Q|\} \}.
 \]
 More specifically, as paths of length $|Q|$ or more induce a (self-)loop,
 we have $\delta(q, a^{|Q|+1}) = s_f$
 if and only if $\delta(q, a^{|Q|}) = s_f$, which yields $S_{|Q| + 1} = S_{|Q|}$
 and so, inductively, $S_i = S_{|Q|}$ for all $i \ge |Q|$.
 Then, for each $i \in \{0, \ldots, |Q|\}$, set $T_i = \delta^{-1}(S_i, c)$
 and decide if there exists a word $w \in \{a,b\}^*$ in $A_{|\{a,b\}} = (\{a,b\}, Q, \delta_{|\{a,b\}})$
 such that $\delta_{|\{a,b\}}(Q, w) \subseteq T_i$.
 By Theorem~\ref{thm:sync_into_subset_waa_in_P}, the last step could be done in polynomial time
 and as we only have to perform this step $|Q| + 1$ many times
 and the sets $T_i$ and $S_i$ are computable in polynomial time, 
 the whole procedure runs in polynomial time.~\qed
\end{proof}

Combining the results of this section, we can give a precise classification
of the complexity landscape for the problem with weakly acyclic input automata
and when the constraint automaton\footnote{Recall that the constraint automaton is a partial automaton, whereas 
the input \mbox{(semi-)auto}maton is always complete.}
has at most two states over a ternary alphabet.

\begin{theoremrep}
\label{thm:classification_waa}
 Let $\mathcal B = (\Sigma, P, \mu, p_0, F)$ be a PDFA. If $|P| \le 1$ or $|P| = 2$
 and $|\Sigma| \le 2$, then $L(\mathcal B)\textsc{-WAA-Constr-Sync}\in \PTIME$.
 For $|P| = 2$ with $|\Sigma| = 3$, up to symmetry by renaming of the letters, 
 $L(\mathcal B)\textsc{-WAA-Constr-Sync}$ is $\NP$-complete precisely
 for the cases listed in Proposition~\ref{prop:np_complete_cases}
 and in $\PTIME$ otherwise.
\end{theoremrep}
\begin{proof}
 By Theorem~\ref{thm:classification_MFCS_paper},
 for at most two states and a binary alphabet, the problem is always in $\PTIME$.
 Up to symmetry,
 for $|P| = 2$ and $|\Sigma| = 3$
 we only have to check the cases listed in Theorem~\ref{thm:classification_MFCS_paper},
 as for the other ones it is in $\PTIME$ 
 for general input automata.
 Except for the cases the cases
 \[
 (a+b)^*c, \quad 
(a+b)^*ca^* \mbox{ and }
(a+b)^*cc^*,
 \]
 these are all listed in Proposition~\ref{prop:np_complete_cases}.
 And for the cases from Theorem~\ref{thm:classification_MFCS_paper}
 not appearing in Proposition~\ref{prop:np_complete_cases}
 written above we have polynomial time solvable
 problems by Proposition~\ref{prop:a_plus_b_c}
 and Proposition~\ref{prop:a_plus_b_c_and_more}.
 Containment in $\NP$ is stated in Theorem~\ref{thm:constraint_sync_weakly_acyclic_in_NP}.~\qed
\end{proof}

\section{Relation to Automata with TTSPL Automaton Graphs}

In~\cite{DBLP:conf/ncma/BruchertseiferF19,DBLP:conf/tamc/BruchertseiferF20} the decision problem related to minimal synchronizing
words was investigated for TTSPL automata. These are automata whose automaton graph, i.e., the multigraph resulting after forgetting about the labels, is a
TTSPL graph, i.e, a two-terminal series-parallel graph with a start and sink node and where self-loops are allowed.

In the context of automata theory, such automata were originally studied in connection
with the size of resulting regular expressions, i.e., motivated by questions
on the descriptional complexity of formal languages~\cite{DBLP:journals/mst/Gulan13}.

Many problems for series-parallel graphs are computationally easy~\cite{DBLP:journals/iandc/Eppstein92}, which partly
motivated the aforementioned studies~\cite{DBLP:conf/ncma/BruchertseiferF19,DBLP:conf/tamc/BruchertseiferF20}. However, from a fixed parameter complexity perspective, for most parameters, synchronization problems
remain hard on the corresponding automata class~\cite{DBLP:conf/ncma/BruchertseiferF19,DBLP:conf/tamc/BruchertseiferF20}.

We will not give all the definitions, but refer the interested reader to the aforementioned papers.
We only mention in passing that TTSPL automata form a proper subclass of the weakly acyclic automata.
Also, by employing a similar construction as used in~\cite[Proposition 4.1]{DBLP:conf/ncma/BruchertseiferF19},
i.e.,
introducing two additional letters, an additional starting state and some auxiliary states
to realize several paths from the start state by a tree-like structure to the starting
states of the paths corresponding to the clauses in the reduction,
we can alter the reduction from Proposition~\ref{prop:np_complete_cases}
to yield a TTSPL graph.  However, we can even do better
and note that for the reductions used in Proposition~\ref{prop:np_complete_cases},
we do not need additional letters, but can realize the branching from the additional starting state 
with two existing letters and use
a third letter to map the additional states to the sink state.  The resulting automaton is a TTSPL automaton,
for example the transitions going directly to the sink state arise out of parallel compositions.
Hence, we can even state the following.

\begin{theorem}
 For the constrained synchronization problem restricted to input automata
 whose automaton graph is a TTSPL graph, 
 we have the same classification
 result for small constraint PDFAs as stated in Theorem~\ref{thm:classification_waa}.
 In particular, we can realize $\NP$-complete constrained problems.
 %
\end{theorem}


\section{Conclusion}

We have investigated the complexity of the constrained synchronization
problem for weakly acyclic input automata. We noticed that in this setting, the problem
is always in $\NP$. In the general setting, it was possible to have $\PSPACE$-complete constrained 
problems, whereas this is no longer possibly in our setting. 
We have investigated the complexities for small constrained automata in the same way as
done in the general case in~\cite{DBLP:conf/mfcs/FernauGHHVW19}.
We found out that certain problems that are $\PSPACE$-complete in general
become $\NP$-complete, whereas others that are $\PSPACE$-complete even become polynomial time solvable.
A similar phenomenon was observed for certain subset synchronization problems
that are all $\PSPACE$-complete in general.

It is natural to continue this investigation for other classes of automata, to find out what
properties are exactly needed to realize $\PSPACE$-complete problems or for what other
classes we only have $\NP$-complete constrained problems, or what are the minimum requirements
on the input automata to realize $\NP$-complete problems.

Also, a complete classification of all possible realizable complexities, a problem orginally posed in~\cite{DBLP:conf/mfcs/FernauGHHVW19}, is still open. Hence, as a first step it would be interesting
to know if for our restricted problem only the complexities $\PTIME$
and $\NP$-complete arise, or if we can realize a constrained problem equivalent
to some $\NP$-intermediate candidate problem.

\smallskip \noindent \footnotesize
\textbf{Acknowledgement.} I thank the anonymous reviewers for noticing some issues
in the proofs of Theorem~\ref{thm:set_transporter_weakly_acyclic} and Proposition~\ref{prop:np_complete_cases} that have been fixed. Also, I thank them for pointing out typos and some unclear formulations.

\bibliographystyle{splncs04}
\bibliography{ms} 
\end{document}